\documentclass[runningheads]{llncs}

\usepackage{graphicx}
\usepackage{hyperref}
\usepackage{amsbsy}
\usepackage{amssymb}
\usepackage{wrapfig}
\usepackage{amsmath}
\usepackage{enumitem}
\usepackage{multirow}
\usepackage{scrextend}
\usepackage{mathtools}
\usepackage{algorithm}
\usepackage[english]{babel}
\usepackage[utf8]{inputenc}
\usepackage[noend]{algpseudocode}
\usepackage{paralist}

\usepackage{tabto}
\usepackage{xfrac}
\usepackage{tikz}
\usetikzlibrary{patterns,quotes,angles}

\newtheorem{observation}{Observation}

\newcommand{\cov}[0]{\operatorname{cov}}

\DeclarePairedDelimiter{\ceil}{\lceil}{\rceil}
\DeclarePairedDelimiter{\floor}{\lfloor}{\rfloor}

\begin{document}

\title{Hedonic Expertise Games}
\titlerunning{Hedonic Expertise Games}

\author{
Bugra Caskurlu, Fatih Erdem Kizilkaya, and Berkehan Ozen}

\authorrunning{B. Caskurlu \and F. E. Kizilkaya \and B. Ozen}

\institute{
TOBB University of Economics and Technology, Ankara, Turkey\\
\email{bcaskurlu@etu.edu.tr},
\email{f.kizilkaya@etu.edu.tr},
\email{b.ozen@etu.edu.tr}}

\maketitle

\begin{abstract}
We consider a team formation setting where agents have varying levels of expertise in a global set of required skills, and teams are ranked with respect to how well the expertise of teammates complement each other. We model this setting as a hedonic game, and we show that this class of games possesses many desirable properties, some of which are as follows: A partition that is Nash stable, core stable and Pareto optimal is always guaranteed to exist. A contractually individually stable partition (and a Nash stable partition in a restricted setting) can be found in polynomial-time. A core stable partition can be approximated within a factor of $1 - \frac{1}{e}$, and this bound is tight unless $\sf P = NP$. We also introduce a larger and relatively general class of games, which we refer to as monotone submodular hedonic games with common ranking property. We show that the above multi-concept existence guarantee also holds for this larger class of games.

\end{abstract}

\section{Introduction}
\label{Section:Intro}

Hedonic games provide a simple formal model for numerous problems, where a set of agents is required to be partitioned into stable coalitions \cite{Hedonic}, such as research group formation \cite{ResearchGroups}, group activity selection \cite{GroupActivitySelection} or task allocation \cite{Swarm} problems. In this paper, we follow this line of research by introducing a model for formation of stable teams. For the ease of understanding, we define our model below using a simple example where students in a classroom needs to form teams for a project assignment.

In our model, a global set of skills and for each agent a level of expertise in each of these skills are given. For instance, the required skills for a class project assignment may be \texttt{(Python, Java, SQL)} where the expertise of two students, say Alice and Bob, in these skills are \texttt{(1, 3, 3)} and \texttt{(3, 3, 1)} respectively. We measure the success of a team by how well the expertise of teammates complement each other. For instance, notice that Alice may compensate the lack of expertise of Bob in \texttt{SQL}, just as Bob may compensate Alice in \texttt{Python}. We say that a coalition's \textit{joint expertise} in some skill is the maximum level of expertise of its members in that skill, and its \textit{joint utility} is the sum of its joint expertise in each skill. Thus, the team formed by Alice and Bob would have a joint expertise of $3$ in each skill, and thus a joint utility of $9$.

We next define the utility functions of the agents. In our classroom example, even if some students do not contribute to their teams as much as their teammates, they will still receive the same grade as all of their teammates. As this is the case in most situations involving teams, we define the utility of agents simply as the joint utility of their coalition. Note that this means each member of a coalition must have the same utility as others.

In the above setting, notice that all agents are better off in the grand coalition. However, in most real-life scenarios there exists a limit on the sizes of coalitions that can be formed due to inherent constrains and/or coordinational problems. For instance, it would not make sense if the whole class formed a single team, in the classroom example above. Therefore, we additionally have an upper bound on the sizes of coalitions possible to form. The above setting can be modeled as a hedonic game, which we refer to as \textit{hedonic expertise games} (HEGs). HEGs naturally model varieties of team contests such as hackathons in which software developers, graphic designers, project managers, and others (often domain experts) collaborate on software projects.

\medskip

Notice that HEGs are a subclass of \textit{hedonic games with common ranking property} (HGCRP) \cite{Partnerships} since all members of a coalition have the same utility. An HGCRP instance simply consists of a set of agents, and a real-valued joint utility function defined over the set of possible coalitions. Note that the joint utility function of HEGs is additionally monotone and submodular. To the best of our knowledge, HGCRP with monotone and submodular joint utility functions has not yet been studied. We also present existence of equilibrium results for this notable subclass of HGCRP, which we refer to as \textit{monotone submodular HGCRP}, that subsumes HEGs.

Common ranking property guarantees the existence of a core stable partition --- the main stability concept based on group deviations in hedonic games --- which can also be found in polynomial-time \cite{Partnerships}. However, since an HGCRP instance requires exponential space in the number of agents whereas an HEG instance is succinctly representable, finding a core stable partition of an HEG instance is not necessarily polynomial-time solvable. Hence, the computational aspects of HEGs require a complete reconsideration. Note that, however, this is not the case for monotone submodular HGCRP, instances of which require exponential space in the number of agents.

There has been a considerable amount of effort put in the literature to design succinctly representable hedonic games (see, for example \cite{Anonymous,Bogo,HedonicNets}). Among those classes of hedonic games, HEGs are most related to $\mathcal{B}$-hedonic games \cite{Survey}. In $\mathcal{B}$-hedonic games, each agent have preferences over other agents which is  then extended to coalitions based on the most preferred agent in the coalition. HEGs might be thought of as a multidimensional generalization of $\mathcal{B}$-hedonic games where each agent has multiple preferences over other agents for each skill. However, each agent's preferences over other agents is identical in our setting, which is not necessarily the case in $\mathcal{B}$-hedonic games. Hence, $\mathcal{B}$-hedonic games are not a subclass of HEGs. $\mathcal{B}$-hedonic games with strict preferences are known to guarantee the existence of a core stable partition \cite{B-preferences}. We do not, however, need to restrict HEGs to strict preferences for such existence guarantees. We refer to \cite{Handbook} for a detailed discussion of the recent hedonic games literature.

\begin{figure}[h]
\centering
\includegraphics[scale=0.33]{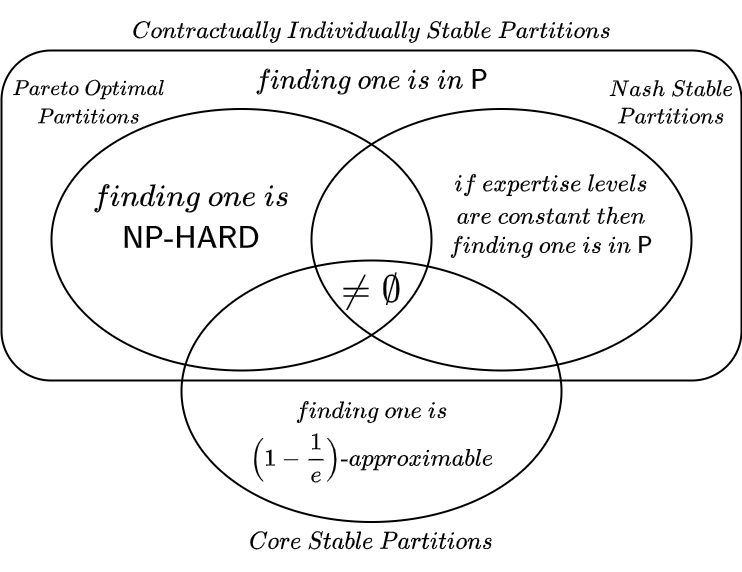}
\caption{The set of core stable, Nash stable, contractually individually stable and Pareto optimal partitions of HEGs are drawn in the above Venn diagram.  The intersection of all of these sets of partitions are guaranteed to be not empty. The computational complexity of finding one such partition in these sets of partitions are stated above.}
\label{Figure:Contributions}
\end{figure}

\noindent \textbf{Contributions and Organization} In Section \ref{Section:Model}, we begin by formally defining HEGs and monotone submodular HGCRP. Then, we define the stability and optimality concepts that we study in the paper.

In Section \ref{Section:Existence}, we prove that a solution that is Nash stable, core stable and Pareto optimal is always guaranteed to exist in monotone submodular HGCRP (and hence in HEGs). The economical interpretation of this existence guarantee is that efficiency need not be sacrificed for the sake of stability with respect to both individual and group based deviations.

In Section \ref{Section:Efficient}, we introduce a decentralized algorithm that finds a Nash stable partition of a given HEG instance. Our procedure terminates in a linear number of moves, if the number of levels of expertise is bounded above by a constant. Note that there exists such a bound for most practical purposes because the expertise in some real-life skill is most commonly measured by a constant number of levels such as \texttt{(0: None, 1: Beginner, 2: Intermediate, 3: Advanced)}. In addition, we show that finding a contractually individually stable partition is polynomial-time solvable in HEGs, even if expertise levels are not bounded.

In Section \ref{Section:Approximation}, we present our main result which is that a $(1 - \frac{1}{e})$-approximate core stable partition of a given HEG instance can be found in polynomial-time and that this is the best approximation ratio achievable, unless $\sf P = NP$. We also show that verifying a core stable partition is intractable, unless $\sf NP = co\text{-}NP$.

In Section \ref{Section:Optimality}, we show that finding a perfect partition, or a socially optimal partition or a Pareto optimal partition is $\sf NP\text{-}HARD$. We also show that verifying a Pareto optimal partition is $\sf coNP$-$\sf COMPLETE$.

In Section \ref{Section:Conclusion}, we present our concluding remarks and discuss some future directions where our contributions might be useful.

\bigskip

The overall picture regarding the multi-concept existence guarantee we give for monotone submodular HGCRP and the related complexity results we give for HEGs is displayed in Figure \ref{Figure:Contributions}.

\section{Model and Background}
\label{Section:Model}

We now formally define \textit{hedonic expertise games} (HEGs). We have a set of \textit{agents} $N$ and a set of \textit{skills} $S$. For each agent $i \in N$, we have a non-negative real-valued \textit{expertise function} $e_i : S \rightarrow \mathbb{R}^+$, where $e_i(s)$ denotes the expertise that agent $i \in N$ has in skill $s \in S$. Lastly, we have an upper bound of $\kappa$ on the sizes of coalitions, i.e., no coalition of size greater than $\kappa$ can be formed. We denote an HEG instance by $\mathcal{G} = (N, S, e, \kappa)$. Moreover, we refer to the subclass of HEGs in which $e_i : S \rightarrow \{0, 1, \ldots, \beta\}$ for all $i \in N$, where $\beta$ is a constant, as $(0, \beta)$-HEGs.

For each coalition $C$, we define the \textit{joint expertise function} $E_C : S \rightarrow \mathbb{R}^+$ where $E_C(s) = \max_{i \in C} e_i(s)$, i.e., $E_C(s)$ is the maximum expertise level that a member of coalition $C$ has in skill $s$. Lastly, we define the \textit{joint utility} of a coalition $C$ as $U(C) = \sum_{s \in S} E_C(s)$, i.e., $U(C)$ is the sum of the maximum expertise that a member of coalition $C$ has in each skill $s \in S$. We define $U(\emptyset) = 0$ as a convention.

A solution of an HEG instance is a partition $\pi$ over the set of agents $N$ where for each coalition  $C \in \pi$ we have $|C| \leq \kappa$. (Throughout the rest of the paper, when we refer to a partition $\pi$, it is implicitly assumed that $|C| \leq \kappa$ for all $C \in \pi$ for the sake of briefness.) We use $\pi(i)$ to denote the coalition containing agent $i \in N$ in partition $\pi$.  We use $u_i(\pi)$ to denote the \textit{utility} of agent $i$ in partition $\pi$ where $u_i(\pi) = U(\pi(i))$, i.e., the utilities of all members of a coalition $C \in \pi$ are the same, and equal to $U(C)$.

The first thing to notice about the above definition is that the joint utility function $U$ is submodular, i.e., for every $X, Y \subseteq N$ with $X \subseteq Y$ and for every $x \in N \setminus Y$, we have that $U(X \cup \{x\}) - U(X) \geq U(Y \cup \{x\}) - U(Y)$. Moreover, the joint utility function $U$ is also monotone, i.e., $U(X) \leq U(Y)$ for all $X \subseteq Y \subseteq N$. We state these properties of the function $U$ in Observation \ref{Observation:Submodular}, the proof of which is given in the appendix.

\begin{observation} \label{Observation:Submodular}
In HEGs, $U$ is a monotone submodular function.
\end{observation}

Recall that HEGs are a subclass of HGCRP, an instance of which is a pair $(N, U)$, where $N$ is a set of agents and $U$ is a joint utility function. Due to Observation \ref{Observation:Submodular}, HEGs are more specifically a subclass of HGCRP with a monotone submodular joint utility function, which have not been studied in the literature to the best of our knowledge. A \textit{monotone submodular HGCRP} instance is a triple $\mathcal{G} = (N, U, \kappa)$, where $N$ is a set of agents, $U$ is a monotone submodular joint utility function, and $\kappa$ is an upper bound on the sizes of coalitions. Notice that, as is the case in HEGs, this game form would be trivial without an upper bound on the size of the coalitions since otherwise all agents would be better off in the grand coalition.

\subsection{Stability \& Optimality}

We now formally define the stability and optimality concepts that we study, in the context of monotone submodular HGCRP.

\medskip

\noindent The main stability concepts based on individual deviations \cite{Bogo} are as follows:

\medskip

\begin{compactitem}
\item A partition $\pi$ is \textit{Nash stable (NS)} if no agent $i \in N$ can benefit from moving to an existing\footnote{Moving to an empty coalition is also permissable, but we can omit this case w.l.o.g. since the joint utility function $U$ is monotone in this setting.} coalition $C \in \pi$ such that $|C| < \kappa$, i.e.,
\begin{compactitem}
\item[--] $U(\pi(i)) \geq U(C \cup \{i\})$
\end{compactitem}
for all $C \in \pi$ such that $|C| < \kappa$.

\medskip

\item A partition $\pi$ is \textit{contractually individually stable (CIS)} if no agent $i \in N$ can benefit from moving to an existing coalition $C \in \pi$ such that $|C| < \kappa$, without making an agent in neither $C$ nor $\pi(i)$ worse off, i.e.,
\begin{compactitem}
\item[--] $U(\pi(i)) \geq U(C \cup \{i\})$ or
\item[--] $U(C) > U(C \cup \{i\})$ or
\item[--] $U(\pi(i)) > U(\pi(i) \setminus \{i\})$
\end{compactitem}
for all $C \in \pi$ such that $|C| < \kappa$.
\end{compactitem}
\medskip

We omitted one other main stability concept based on individual deviations, known as individually stable partitions, since a partition is NS if and only if it is an individually stable partition in monotone submodular HGCRP.

\medskip

\noindent The main stability concept based on group deviations \cite{Core}, and its approximate adaptation is as follows:

\medskip

\begin{compactitem}
\item A coalition $C$ is said to \textit{block} $\pi$, if $U(C) > u_i(\pi)$ for all agents $i \in C$, i.e., any agent $i \in C$ is  better off in $C$ than she is in her coalition $\pi(i)$. A partition $\pi$ is \textit{core stable (CS)} if no coalition blocks $\pi$.

\medskip

\item Similarly, a coalition $C$ is said to \textit{$\alpha$-approximately block} $\pi$ where $\alpha \leq 1$, if $\alpha \cdot U(C) > u_i(\pi)$ for all agents $i \in C$. Similarly, a partition $\pi$ is \textit{$\alpha$-approximate CS} if no coalition $\alpha$-approximately blocks $\pi$. Note that a $1$-approximate CS partition is simply a CS partition.
\end{compactitem}

\medskip

We now define the following notation which comes in handy with the above stability concepts. For a partition $\pi$ and a coalition $C \notin \pi$, we define $\pi_C$ as the partition induced on $\pi$ by $C$, i.e., $\pi_C$ is the partition that would arise if the agents in $C$ collectively deviated from $\pi$ to form coalition $C$, i.e., $\pi_C(i) = C$ for all $i \in C$, and $\pi_C(j) = \pi(j) \setminus C$ for all $j \in N \setminus C$. Notice that if a partition $\pi$ is not NS then there exists an agent $i \in N$ and a coalition $C \in \pi$ such that $u_i(\pi_{C \cup \{i\}}) > u_i(\pi)$ and $|C| < \kappa$. Also notice that if coalition $C$ blocks partition $\pi$ then $u_i(\pi_C) > u_i(\pi)$, for all agents $i \in C$.

\medskip

\noindent The main optimality concepts are as follows:

\medskip

\begin{compactitem}
\item A partition $\pi$ is \textit{perfect} if all agents are in their most preferred coalition.

\medskip

\item The \textit{social welfare} $W(\pi)$ of a partition $\pi$ is defined as the sum of the utilities of all the agents, i.e., $W(\pi) = \sum_{i \in N} u_i(\pi)$. A \textit{socially optimal (SO)} solution is a partition for which the social welfare is maximized.

\medskip

\item A partition $\pi'$ \textit{Pareto dominates} a partition $\pi$ if $u_i(\pi') \ge u_i(\pi)$ for all agents $i \in N$, and there exists some agent $i$ for which the inequality is strict. A partition $\pi$ is said to be \textit{Pareto optimal (PO)} if no partition $\pi'$ Pareto dominates $\pi$.
\end{compactitem}

\medskip

Observe that a perfect partition is necessarily socially optimal, and a socially optimal partition is necessarily Pareto optimal. Though a socially optimal partition (and thus a Pareto optimal partition) is guaranteed to exist in hedonic games, a perfect partition does not necessarily exist in HEGs (and thus in monotone submodular HGCRP).

\section{Existence Guarantees}
\label{Section:Existence}

Common ranking property have been long known for guaranteeing the existence of a CS partition in hedonic games via a simple greedy algorithm \cite{Partnerships}. Moreover, existence of a partition which is both CS and PO in HGCRP is recently proven by presenting an asymmetric and transitive relation $\psi$ defined over the set of partitions, where a maximal partition with respect to $\psi$ is  both CS and PO \cite{HGCRP}. This is established by applying the potential function argument twice as follows.

Given a partition $\pi$, $\psi(\pi)$ is defined as the sequence of the utilities of all the agents in a \textit{non-increasing order}. Then, it is shown that \textbf{(CS)} and \textbf{(PO)} given below hold for any partition $\pi$, where we use $\rhd$ to denote ``lexicographically greater than''. (We will be also using $\unrhd$ to denote ``lexicographically greater than or equal to''.)

\medskip

\begin{compactitem}
\item[\textbf{(CS)}] If there exists a coalition $C$ that blocks $\pi$ then $\psi(\pi_C) \rhd \psi(\pi)$.
\item[\textbf{(PO)}] If there exists a partition $\pi'$ that Pareto dominates $\pi$ then $\psi(\pi') \rhd \psi(\pi)$.
\end{compactitem}

\medskip

On the other hand, an HGCRP instance does not necessarily possess a NS partition as can be seen from Example \ref{Example:HGCRP}.

\begin{example} \label{Example:HGCRP}
Consider the HGCRP instance $\mathcal{G} = (N, U)$, where $N = \{1, 2\}$, and $U$ is defined as $U(\{1\}) = 1$, $U(\{1, 2\}) = 2$ and $U(\{2\}) = 3$. Notice that HGCRP instance $\mathcal{G}$ does not possess a NS partition.
\end{example}

In this section, we show that monotone submodular HGCRP do not only admit a NS partition but also a NS, CS and PO partition; which improves upon the aforementioned existence guarantee in HGCRP.

\begin{theorem} \label{Theorem:Multiconcept}
In monotone submodular HGCRP, a partition that is NS, CS and PO is always guaranteed to exist.
\end{theorem}

In order to prove Theorem \ref{Theorem:Multiconcept}, notice that we only need to include a third potential function argument which establishes that a maximal partition with respect to $\psi$ is also NS, due to \textbf{(CS)} and \textbf{(PO)}.

\begin{lemma} \label{Lemma:NashStable}
Given a partition $\pi$ in monotone submodular HGCRP, if there exists an agent $i$ which benefits from moving to an existing coalition $C \in \pi$ such that $|C| < \kappa$ (i.e., $u_i(\pi_{C \cup \{i\}}) > u_i(\pi)$), then $\psi(\pi_{C \cup \{i\}}) \rhd \psi(\pi)$.
\end{lemma}

\begin{proof}
Suppose that a monotone submodular HGCRP instance $\mathcal{G} = (N, U, \kappa)$ is given along with some partition $\pi$. We denote the $i^{th}$ element in $\psi(\pi)$ by $\psi_i(\pi)$. We rename agents such that $\psi_i(\pi) = u_i(\pi)$, i.e., $u_1(\pi) \geq u_2(\pi) \geq \ldots \geq u_{|N|}(\pi)$. Moreover, we assume w.l.o.g. that the utilities of agents that are in the same coalition in $\pi$ are listed consecutively in $\psi(\pi)$.

Suppose that there exists an agent $i \in N$ which benefits from moving to an existing coalition $C \in \pi$ such that $|C| < \kappa$, i.e., $U(C \cup \{i\}) \geq U(\pi(i))$. Furthermore, assume w.l.o.g. that the utility of agent $i$ precedes the utility of other agents in $\pi(i)$ in the ordering of $\psi(\pi)$.

Let $\psi'(\pi_{C \cup \{i\}})$ be a permutation of $\psi(\pi_{C \cup \{i\}})$ such that $\psi'_j(\pi_{C \cup \{i\}}) = u_j(\pi_{C \cup \{i\}})$ for all agents $j \in N$. Note that $\psi'_j(\pi_{C \cup \{i\}})$ and $\psi_j(\pi)$ are the respective utilities of the same agent $j$. Moreover, note that since  $\psi(\pi_{C \cup \{i\}})$ is the same sequence as $\psi'(\pi_{C \cup \{i\}})$ but sorted in non-increasing order, we have $\psi(\pi_{C \cup \{i\}}) \unrhd \psi'(\pi_{C \cup \{i\}})$.

We know that $\psi'(\pi_{C \cup \{i\}}) \neq \psi(\pi)$ since their $i^{th}$ elements differ. This means that either $\psi'(\pi_{C \cup \{i\}}) \rhd \psi(\pi)$ or $\psi(\pi) \rhd \psi'(\pi_{C \cup \{i\}})$. Assume for the sake of contradiction that $\psi(\pi) \rhd \psi'(\pi_{C \cup \{i\}})$. Let $j$ be the agent with the smallest index such that $\psi_j(\pi) >  \psi'_j(\pi_{C \cup \{i\}})$, i.e., $u_j(\pi) > u_j(\pi_{C \cup \{i\}})$. Note that there must exists such $j \in N$. This implies that $j \in \pi(i)$ since only coalition whose joint utility may decrease is $\pi(i)$ after agent $i$ moves to coalition $C$ in partition $\pi$. However, recall that $u_i(\pi)$ precedes $u_j(\pi)$ in the ordering of $\psi(\pi)$, i.e., there exists an agent $i$ where $i < j$ such that $\psi_i(\pi) <  \psi'_i(\pi_{C \cup \{i\}})$, which is contradictory. Therefore, we have $\psi'(\pi_{C \cup \{i\}}) \rhd \psi(\pi)$.

Recall that $\psi(\pi_{C \cup \{i\}}) \unrhd \psi'(\pi_{C \cup \{i\}})$. Then, $\psi(\pi_{C \cup \{i\}}) \unrhd \psi'(\pi_{C \cup \{i\}}) \rhd \psi(\pi)$. Since lexicographic order is transitive, this means that $\psi(\pi_{C \cup \{i\}}) \rhd \psi(\pi)$, which completes our proof. \qed
\end{proof}

Since HEGs is a subclass of monotone submodular HGCRP, a partition that is NS, CS and PO is also guaranteed to exist in HEGs by Theorem \ref{Theorem:Multiconcept}.

\begin{corollary} \label{Theorem:Multiconcept}
In HEGs, a partition that is NS, CS and PO is always guaranteed to exist.
\end{corollary}

\section{Efficiently Computable Stable Partitions}
\label{Section:Efficient}

Recall that since a monotone submodular HGCRP instance requires exponential space in the number of agents, finding a solution in polynomial-time does not necessarily indicate an efficient algorithm in the practical sense. Therefore, we focus on HEGs with regard to computational aspects.

We first present a decentralized algorithm for finding a NS partition of a given $(0,\beta)$-HEG instance. Our algorithm accompanies a restricted version of \textit{better response dynamics}, i.e., while the current partition $\pi$ is not NS, an agent $i$ moves to an existing coalition $C \in \pi$ such that $|C| < \kappa$ and $u_i(\pi_{C \cup \{i\}}) > u_i(\pi)$, which we refer to as a \textit{better response} of agent $i$ in partition $\pi$.

Due to Lemma \ref{Lemma:NashStable}, we know that better response dynamics is guaranteed to converge to a NS partition in HEGs. However, it is not clear how many moves this would require in the worst-case. For the ease of analysis,  we force a natural restriction on better response dynamics, under which we show that a linear number of moves suffice to find a NS partition. We refer to this restricted class as \textit{imitative better response dynamics}, which we describe below.

\medskip

\begin{compactitem}
\item Given a partition $\pi$, suppose that an agent $i$ benefits from moving to an existing coalition $C \in \pi$ such that $|C| < \kappa$, i.e., $U(C \cup \{i\}) > U(\pi(i))$.

\item Suppose that agent $i$ takes the above better response. If $|C \cup \{i\}| < \kappa$, then notice that another agent $i' \in \pi(i) \setminus \{i\}$ also benefits from moving to $C \cup \{i\}$, since $U(C \cup \{i\}) > U(\pi(i)) \geq U(\pi(i) \setminus \{i\})$.

\item That is, if the size of the coalition that the last agent $i$ has moved to did not reach the upper bound of $\kappa$, then an agent $i' \in \pi(i) \setminus \{i\}$ simply \textit{``imitates''} the last agent $i$ by moving to the same coalition. Otherwise, an arbitrary agent takes a better response.
\end{compactitem}

\medskip

We now show that above procedure starting from any partition $\pi$ where each coalition $C \in \pi$ has a size of exactly $\kappa$ (except maybe for one coalition) converge to a NS partition in $O(|N| \cdot |S|)$ moves, in $(0,\beta)$-HEGs. Notice that this is significant since we can say that even boundedly-rational agents that imitate the last agent if possible, and take the first beneficial move otherwise, will converge to a NS partition quickly.

\begin{theorem} \label{Theorem:Poly:NS}
In $(0,\beta)$-HEGs, a NS partition can be found in polynomial-time via imitative better response dynamics in $O(|N| \cdot |S|)$ moves.
\end{theorem}

\begin{proof}
Suppose that a $(0,\beta)$-HEG instance $\mathcal{G} = (N, S, e, \kappa)$ is given. We begin by constructing a random partition $\pi$ where each coalition $C \in \pi$ has a size of exactly $\kappa$, except maybe for one coalition. We refer to the coalitions in $\pi$ whose size is exactly $\kappa$ as $C_1, \ldots, C_{\floor{\sfrac{|N|}{\kappa}}}$. Notice that if these are the only coalitions in $\pi$ then we are done since no agent can move to another coalition in $\pi$. Hence, we assume w.l.o.g. that there is another coalition $L \in \pi$ such that $|L| < \kappa$ which consists of the leftover agents from those coalitions whose sizes are exactly $\kappa$. We now consider imitative better response dynamics starting from partition $\pi$.

Suppose that $\pi$ is not NS. Then, there exists an agent $j \in C_i$ which benefits from moving to coalition $L$, i.e., $U(L \cup \{j\}) > U(C_i)$. Suppose that agent $j$ takes this better response. Note that then $\kappa - |L| - 1$ agents in $C_i \setminus \{j\}$ will imitate agent $j$ by moving to the same coalition. Notice that after these $\kappa - |L|$ moves, the resulting partition will still consist of $\floor{\sfrac{|N|}{\kappa}}$ coalitions whose sizes are exactly $\kappa$, and an additional coalition that consists of the leftover agents. We exploit this structure as follows:

\medskip

\begin{compactitem}
\item Let $C'_i$ denote the resulting coalition after agent $j$ and other $\kappa - |L| - 1$ agents in $C_i \setminus \{j\}$ move to coalition $L$, i.e., $C'_i = L \cup \{j\} \cup K$ where $K \subseteq C_i \setminus \{j\}$ is an arbitrary subset of agents of size $\kappa - |L| - 1$.

\medskip

\item Let $L'$ denote the remaining coalition after agent $j$ and other $\kappa - |L| - 1$ agents in $C_i \setminus \{j\}$ move to coalition $L$, i.e., $L' = C_i \setminus (K \cup \{j\})$.
\end{compactitem}

\medskip

Notice that we can obtain the resulting partition, say $\pi'$, after these $\kappa - |L|$ moves by updating $C_i$ as $C'_i$ and $L$ as $L'$ in partition $\pi$. Moreover, notice that $U(C'_i) = U(L \cup \{j\} \cup K) \geq U(L \cup \{j\}) > U(C_i)$, which means the joint utility of the coalition which we refer to as $C_i$ is strictly greater in $\pi'$ than in $\pi$. Since $U(C_i)$ is an integer between $0$ and $\beta \cdot |S|$, this means the number of moves is bounded by $\beta \cdot |S| \cdot \floor{\sfrac{|N|}{\kappa}} \cdot (\kappa - |L|) = O(|N| \cdot |S|)$, which finishes our proof. \qed
\end{proof}

We next present a polynomial-time algorithm for finding a CIS partition in HEGs, not only in $(0, \beta)$-HEGs unlike our previous result.

\begin{theorem} \label{Theorem:Poly:CIS}
In HEGs, a CIS partition can be found in polynomial-time.
\end{theorem}

\begin{proof} Suppose that an HEG instance $\mathcal{G} = (N, S, e, \kappa)$ is given. Exactly as in Theorem \ref{Theorem:Poly:NS}, we begin with a partition $\pi = (C_1, \ldots, C_{\floor{\sfrac{|N|}{\kappa}}}, L)$ where $|C_i| = \kappa$ for all $C_i \in \pi$ and $|L| < \kappa$. Note that if $L = \emptyset$ (i.e. $\kappa$ divides $|N|$) then we are already done. Hence, we assume w.l.o.g. that $L \neq \emptyset$.

Recognizing which agents of a coalition are (or would be) ``critical'' lies in the heart of our proof.  We say that an agent $i$ is \textit{critical} for a coalition $C$ if there exists a skill $s \in S$ such that $e_i(s) > E_{C \setminus \{i\}}(s)$. Note that an agent $i \in C$ is critical for coalition $C$ if and only if  $U(C) > U(C \setminus \{i\})$. This means that if each agent $i \in C$ is critical for $C$, then no agent in $C$ can leave coalition $C$ without making an agent in $C$ worse off.

\medskip

Notice that $\pi$ is not CIS if and only if there exists an agent $j \in C_i$ such that\footnote{And also $U(L \cup \{j\}) \geq U(L)$ which trivially holds in HEGs.}:

\medskip

\begin{compactitem}
\item $U(L \cup \{j\}) > U(C_i)$

(which means there exists a skill $s \in S$ such that $E_L(s) > E_{C_i}(s)$, which then implies that there exists a critical agent $j' \in L$ for coalition $C_i$),

\medskip

\item $U(C_i \setminus \{j\}) \geq U(C_i)$

(which means that agent $j$ is not critical for $C_i$),
\end{compactitem}

\medskip

If this is the case, then we update $C_i$ and $L$ as follows:

\medskip

\begin{compactitem}
\item Let $C'_i = (C_i \setminus \{j\}) \cup \{j'\}$.
\item Let $L' = (L \setminus \{j'\}) \cup \{j\}$.
\end{compactitem}

\medskip

Let $\gamma(C)$ denote the number of critical agents for coalition $C$ that are also in coalition $C$. Since $j'$ is critical for $C_i$ whereas $j$ is not, we have $\gamma(C'_i) > \gamma(C_i)$. Therefore, if we update $C_i$ as $C'_i$, $L$ as $L'$ and repeat the above procedure, we will eventually reach a CIS partition. We now only need to show that the number of iterations will be polynomial.

Notice that $E_{C'_i}(s) \geq E_{C_i}(s)$ for all $s \in S$. Therefore, no matter how many iterations have passed, agent $j$ cannot ever become critical for coalition $C_i$. However, for $j$ to be able to return back to $C_i$, she must be critical for $C_i$. Hence, $j$ cannot ever return back to $C_i$. This means that the number of iterations is bounded by $\floor{\sfrac{|N|}{\kappa}} \cdot |N|$, which finishes our proof. \qed
\end{proof}

\section{Approximating Core Stable Partitions}
\label{Section:Approximation}

We devote this section to show that finding a $(1-\frac{1}{e})$-approximate CS partition of an HEG instance is polynomial-time solvable, and this is the best possible approximation ratio achievable, unless $\sf P = NP$.

\medskip

Initially, we study the problem of finding a coalition with maximum joint utility, which we formally specify as follows.

\bigskip

$\sf MAXIMUM$-$\sf JOINT$-$\sf UTILITY$ $=$ \textit{``Given an HEG instance $\mathcal{G} = (N, S, e, \kappa)$, find a subset of agents $C^* \subseteq N$ which maximizes $U(C^*)$ such that $|C^*| \leq \kappa$.''}

\bigskip

We now show that the above problem is $(1 - \frac{1}{e})$-inapproximable even for $(0,1)$-HEGs, unless $\sf P = NP$, which we then use to show that a $(1 - \frac{1}{e} + \epsilon)$-approximate CS partition cannot be found in polynomial-time for any $\epsilon > 0$, unless $\sf P = NP$.

\begin{lemma} \label{Lemma:Inapproximability:MaximumJointUtility}
In $(0,1)$-HEGs, the $\sf MAXIMUM$-$\sf JOINT$-$\sf UTILITY$ problem is inapproximable within better than ratio of $1 - \frac{1}{e}$, unless $\sf P = NP$.
\end{lemma}

\begin{proof}
We give an approximation preserving $S$-reduction \cite{ApproximationPreservingReductions} from $\sf MAXIMUM$-$\sf COVERAGE$ problem, which is known to be inapproximable within better than $1 - \frac{1}{e}$, unless $\sf P = NP$ \cite{ApproximationThreshold}. An instance of $\sf MAXIMUM$-$\sf COVERAGE$ problem consists of a universe $\mathcal{U} = \{1, \ldots, m\}$, a family $\mathcal{S} = \{\mathcal{S}_1, \ldots, \mathcal{S}_n\}$ of subsets of $\mathcal{U}$, an integer $k$, and the objective of finding a subset $\mathcal{C} \subseteq \mathcal{S}$ such that $|\mathcal{C}| \leq k$ which maximizes $\cov(\mathcal{C}) = |\cup_{\mathcal{S}_i \in \mathcal{C}} \mathcal{S}_i|$. Given a $\sf MAXIMUM$-$\sf COVERAGE$ instance $\mathcal{I} = (\mathcal{U}, \mathcal{S}, k)$, we build a $(0, 1)$-HEG instance $\mathcal{G_I} = (N, S, e, \kappa)$ as follows:

\medskip

\begin{compactitem}
\item[--] Universe $\mathcal{U}$ corresponds to the set of skills $S$.
\item[--] Each subset $\mathcal{S}_i \in \mathcal{S}$ corresponds to an agent $i \in N$, whose expertise function is defined for each skill $s \in S$ as $e_i(s) = 1$ if $s \in \mathcal{S}_i$, and $e_i(s) = 0$ otherwise.

\item[--] Finally, $k$ corresponds to the upper bound $\kappa$ on the sizes of coalitions.
\end{compactitem}

\medskip

Notice that a coalition $C$ (i.e., a subset of the set of agents $N$ such that $|C| \leq \kappa$) corresponds to a subset $\mathcal{C} \subseteq \mathcal{S}$ such that $U(C) = \cov(\mathcal{C})$ and $|\mathcal{C}| \leq k$; and the reverse also holds, which completes our $S$-reduction. \qed
\end{proof}

\begin{theorem} \label{Theorem:Inapproximability:CoreStable}
In $(0,1)$-HEGs, a $(1 - \frac{1}{e} + \epsilon)$-approximate CS partition cannot be found in polynomial-time for any constant $\epsilon > 0$, unless $\sf P = NP$.
\end{theorem}

\begin{proof}
For some $\epsilon > 0$, suppose that a $(1 - \frac{1}{e} + \epsilon)$-approximate CS partition $\pi$ of a given $(0,1)$-HEG instance $\mathcal{G} = (N, S, e, \kappa)$ can be found in polynomial-time. Let $C^*$ be a coalition in $\mathcal{G}$ with maximum joint utility. Then, there exists an agent $i^* \in C^*$ such that $(1 - \frac{1}{e} + \epsilon) \cdot U(C^*) \leq u_{i^*}(\pi)$, because otherwise $C^*$ would $(1 - \frac{1}{e} + \epsilon)$-approximately block $\pi$.

Let $C \in \pi$ be a coalition such that $U(C) \geq U(C')$ for all $C' \in \pi$. Note that $C$ can be found in polynomial-time. Notice that $(1 - \frac{1}{e} + \epsilon) \cdot U(C^*) \leq u_{i^*}(\pi)\leq U(C)$. This means that we could devise a $(1 - \frac{1}{e} + \epsilon)$-approximation algorithm for $\sf MAXIMUM$-$\sf JOINT$-$\sf UTILITY$ problem  by simply returning $C$. Unless $\sf P = NP$, this creates a contradiction by Lemma \ref{Lemma:Inapproximability:MaximumJointUtility}, which finishes our proof.\qed
\end{proof}

Theorem \ref{Theorem:Inapproximability:CoreStable} also has the following interesting implication under the widely believed assumption that $\sf NP \neq co\text{-}NP$.

\begin{theorem} \label{Theorem:Verification:CoreStable}
In $(0, 1)$-HEGs, it is not possible to verify whether a given partition is CS or not in polynomial-time, unless $\sf NP = co\text{-}NP$.
\end{theorem}

\begin{proof}
We first define a local search problem on solutions of a given $(0,1)$-HEG instance $\mathcal{G} = (N, S, e, \kappa)$ as follows:

\medskip

\begin{compactitem}
\item[--] For all partitions $\pi$, we define all partitions $\pi_C$ where a coalition $C \subseteq N$ blocks partition $\pi$ as the neighbors of partition $\pi$.

\item[--] We define a binary relation $\succ$ which is to be locally maximized, where $\pi_C \succ \pi$ for all of those neighbors $\pi_C$ of partition $\pi$.
\end{compactitem}

\medskip

Due to the potential function $\psi$ we have given in Section \ref{Section:Existence}, notice that finding a locally maximum solution with respect to the binary relation $\succ$ and the neighborhood relation defined above, is exactly the same problem as finding a CS partition. This means that finding a CS partition of a given $(0,1)$-HEG instance is a local search problem by definition.

Local search problems in which a solution can be verified to be optimal in polynomial-time are called polynomial local search problems, whose complexity class is denoted as $\sf PLS$.
If there exists a problem in $\sf PLS$ which is also $\sf NP\text{-}HARD$, then it is known that $\sf NP = co\text{-}NP$ \cite{PLS}.

Due to Theorem \ref{Theorem:Inapproximability:CoreStable}, we know that the local search problem of finding a CS partition of a given $(0,1)$-HEG instance is $\sf NP\text{-}HARD$. This means that it cannot be also in $\sf PLS$, unless $\sf NP = co\text{-}NP$. Then, it cannot be possible to verify if a given partition is CS or not in polynomial-time by definition unless $\sf NP = co\text{-}NP$, which finishes our proof. \qed

\end{proof}

We now show that $\sf MAXIMUM$-$\sf JOINT$-$\sf UTILITY$ problem is approximable within a ratio of $1 - \frac{1}{e}$ by the standard greedy algorithm: \textit{``Begin with an empty coalition $C$, and then greedily add the agent that increase the joint utility of $C$ the most, until reaching the upper bound of $\kappa$''}. We then use the above algorithm as a subroutine to find a $(1-\frac{1}{e})$-approximate CS partition in HEGs.

\begin{lemma} \label{Lemma:Appromixation:MaximumJointUtility}
In HEGs, $\sf MAXIMUM$-$\sf JOINT$-$\sf UTILITY$ problem is $(1 - \frac{1}{e})$-approximable by the standard greedy algorithm.
\end{lemma}

\begin{proof}
Recall that the joint utility function $U$ is monotone and submodular as given in Observation \ref{Observation:Submodular}. Due to the upper bound on the sizes of coalitions, this means that $\sf MAXIMUM$-$\sf JOINT$-$\sf UTILITY$ is simply a problem of maximizing a monotone submodular function subject to a cardinality constraint, which is $(1 - \frac{1}{e})$-approximable by the standard greedy algorithm \cite{Submodular}. \qed
\end{proof}

\begin{theorem} \label{Theorem:Approximation:CoreStable}
In HEGs, a $(1 - \frac{1}{e})$-approximate CS partition can be found in polynomial-time.
\end{theorem}

\begin{proof}
Suppose that an HEG instance $\mathcal{G} = (N, S, e, \kappa)$ is given. Let $C^*$ be a coalition of $\mathcal{G}$ with maximum joint utility. Note that running the standard greedy algorithm will return a coalition $C$ such that $U(C) \geq (1 - \frac{1}{e}) \cdot U(C^*)$ due to Lemma \ref{Lemma:Appromixation:MaximumJointUtility}. Since no agent can have a strictly greater utility than $U(C^*)$, notice that if $\pi$ is a partition where $C \in \pi$ then no agent $i \in C$ can participate in a coalition that $(1-\frac{1}{e})$-approximately blocks $\pi$. Thus, by forming coalition $C$, we can reduce the problem of finding a $(1-\frac{1}{e})$-approximate CS partition in $\mathcal{G}$ into one of finding a $(1-\frac{1}{e})$-approximate CS partition in $\mathcal{G'} = (N \setminus C, S, e, \kappa)$. This means that we can build a $(1-\frac{1}{e})$-approximate CS partition in polynomial-time by repeatedly running the standard greedy algorithm, which finishes our proof.  \qed
\end{proof}

\section{Intractability of Computing Optimal Partitions}
\label{Section:Optimality}

We show that finding a perfect, SO or PO partition even in $(0,1)$-HEGs is intractable, and so is even verifying if a partition is PO.

\medskip

For any subclass of hedonic games, if deciding whether there exists a perfect partition is $\sf NP\text{-}HARD$, and if a given partition can be verified to be perfect in polynomial-time, then it is known that finding a PO or SO partition is also $\sf NP\text{-}HARD$ in the same subclass of hedonic games \cite{Pareto}. However, we cannot directly use this method since it is not clear how we can check efficiently whether a given partition is perfect in HEGs. The trick is to construct special HEG instances where a partition can be verified to be perfect easily, throughout the reductions.

\begin{theorem} \label{Theorem:Hardness}
In $(0,1)$-HEGs: $(i)$ deciding if a perfect partition exists is $\sf NP\text{-}HARD$, $(ii)$ finding a SO partition is $\sf NP\text{-}HARD$, $(iii)$ finding a PO partition is $\sf NP\text{-}HARD$ and $(iv)$ verifying whether a partition is PO is $\sf coNP$-$\sf COMPLETE$.
\end{theorem}

\begin{proof}
All of the proofs for $(i)$, $(ii)$, $(iii)$ and $(iv)$ are via a polynomial-time mapping reduction from $\sf SET$-$\sf COVER$ problem, in which we are given a universe $\mathcal{U} = \{1, \ldots, m\}$ and a family $\mathcal{S} = \{\mathcal{S}_1, \ldots, \mathcal{S}_n\}$ of subsets of $\mathcal{U}$ along with a positive integer $k$; and then, we are required to decide whether there exists a \textit{set cover} $C \subseteq \mathcal{S}$ whose size is at most $k$, i.e., we need to have $|C| \leq k$ and $\cup_{\mathcal{S}_i \in C} S_i = \mathcal{U}$. Given an instance $\mathcal{I} = (\mathcal{U, S}, k)$ of $\sf SET$-$\sf COVER$, we build a $(0, 1)$-HEG instance $\mathcal{G_I} = (N, S, e, \kappa)$ as follows:

\medskip

\begin{compactitem}
\item[--] Universe $\mathcal{U}$ corresponds to the set of skills $S$.
\item[--] Each subset $\mathcal{S}_i \in \mathcal{S}$ corresponds to an agent $i \in N$, whose expertise function is defined for each skill $s \in S$ as $e_i(s) = 1$ if $s \in \mathcal{S}_i$, and $e_i(s) = 0$ otherwise.

\item[--] Let $x = \ceil{\frac{n-k}{k-1}}$ and $X = \{n+1, \ldots, n+x\}$ where each $i \in X$ corresponds to an agent $i \in N$ such that $e_i(s) = 1$ for all $s \in S$.

\item[--] Finally, $k$ corresponds to the upper bound $\kappa$ on the sizes of coalitions.
\end{compactitem}

\medskip

Notice that $|N| = n + x$ and $\ceil{\frac{n+x}{k}} = x + 1$. Therefore, a partition $\pi$ of $\mathcal{G_I}$ contains at least $x + 1$ coalitions. Hence, there exists a coalition $C \in \pi$ which does not contain any agent in $X$. Thus, we can map back any partition $\pi$ of $\mathcal{G_I}$ to a feasible solution $\mathcal{C}$ of $\mathcal{I}$ by returning coalition $C$. Notice that coalition $C$ corresponds to a set cover $\mathcal{C}$ of $\mathcal{I}$ if and only if $U(C) = m$.

\paragraph{$(i)$} Since all the agents can attain the utility of $m$ by participating in a coalition with an agent in $X$, a partition $\pi$ is perfect if and only if $u_i(\pi) = m$ for all $i \in N$. Notice that such an allocation $\pi$ exists if and only if there exists a set cover $\mathcal{C}$ of $\mathcal{I}$ whose size is less than $k$. Therefore, deciding whether there exists a perfect partition in $(0, 1)$-HEGs is $\sf NP$-$\sf HARD$.

\paragraph{The rest of our results are proven via the same construction as above but by exploiting a reduction from the problem given in $(i)$.}

\paragraph{$(ii)$} Notice that if there exists a perfect partition, then any SO partition needs to be also perfect. If we could find a SO partition $\pi$ of $\mathcal{G_I}$ in polynomial-time, then we could decide if there exists a perfect partition in polynomial-time by simply checking whether $u_i(\pi) = m$ for all $i \in N$. However, this would be contradictory. Therefore, finding a SO partition is also $\sf NP$-$\sf HARD$ in $(0,1)$-HEGs.

\paragraph{$(iii)$} Assume for the sake of contradiction that we can find a PO partition $\pi$ of $\mathcal{G_I}$ in polynomial-time. Suppose that $\pi$ is not a perfect partition. Then, there cannot exist a perfect partition $\pi^*$ of $\mathcal{G_I}$ because otherwise $\pi^*$ would have Pareto dominate $\pi$ by definition. Then, we could decide if there exists a perfect partition by checking whether $u_i(\pi) = m$ for all $i \in N$,  which would be contradictory. Therefore, finding a PO partition is also $\sf NP$-$\sf HARD$ in $(0,1)$-HEGs.

\paragraph{$(iv)$} Let $\pi = (X_1, \ldots, X_x, C)$ be a partition of $\mathcal{G_I}$ such that $X \cap C = \emptyset$, $U(C) < m$ and $n + i \in X_i$ for all $i$. We show that there exists a perfect partition of $\mathcal{G_I}$ if and only if $\pi$ is not PO. Notice that $\pi$ is not a perfect partition, since $U(C) < m$. Therefore, it is clear that if $\pi$ is PO then there does not exist a perfect partition.

Now suppose that $\pi$ is not PO. Then, there exists a partition, say $\pi'$, that Pareto dominates $\pi$. Recall that there must be a coalition $C' \in \pi'$ such that $C' \cap X = \emptyset$. Note that $C' \neq C$ since otherwise no agent would be better off in $\pi'$ with respect to $\pi$. Moreover, since all agents except those in $C$ had a utility of $m$ in $\pi$, some agents would get worse off in $\pi'$, unless $U(C') = m$. Recall that this implies the existence of a perfect partition, and thus, we are done.

Finally, note that verifying a whether a given partition $\pi$ is PO in $\sf coNP$, since a partition $\pi'$ that Pareto dominates $\pi$ is a counterexample that is verifiable in polynomial-time. Therefore, verifying a PO partition is $\sf coNP$-$\sf COMPLETE$. \qed
\end{proof}

\section{Concluding Remarks and Future Directions}
\label{Section:Conclusion}

In this paper, we investigated computational aspects of HEGs and  we concluded that stable solutions based on individual deviations (namely NS partitions if the level of expertise is bounded by a constant, and CIS partitions in general) can be computed efficiently, whereas stable solutions based on group deviations (namely CS partitions) can be approximated within a factor of $1 - \frac{1}{e} \approx 0.632$. Moreover, we indicated that the existence guarantees given in HEGs arise from the fact that HEGs is a subclass of much more general class of hedonic games, which we referred to as monotone submodular HGCRP in the paper.

Due to reflecting diminishing returns property, we believe that monotone submodular HGCRP are likely to have some other applications, where our contributions in HEGs might be useful. For example, consider the following subclass of HGCRP which is played on a given undirected weighted graph $G = (V,E)$ where
\medskip
\begin{compactitem}
\item the vertices correspond to the agents, and
\item the joint utility of a coalition $C$ is the total weight of the incident edges of the corresponding vertices of $C$ in $G$.
\end{compactitem}
\medskip
We also need a upper bound of $\kappa$ on the sizes of coalitions since the above joint utility function is monotone. Moreover, the joint utility function is submodular. This is simply because one can build a corresponding HEG instance where each edge $e \in E$ corresponds to a skill $s$ such that the expertise of each agent $i$ on skill $s$ is the weight of edge $e$ if $e$ is incident to the corresponding vertex of $i$, and 0 otherwise. Therefore, the above class of hedonic games, which we refer to as \textit{hedonic vertex cover games} (HVCGs), is a subclass of HEGs, and hence monotone submodular HGCRP.

Notice that $\sf MAXIMUM$-$\sf JOINT$-$\sf UTILITY$ problem for a given HVCG instance can be approximated within $\sfrac{3}{4}$ by simply finding a $\sfrac{3}{4}$-approximate maximum vertex cover of $G$ with a cardinality constraint of $\kappa$, which can be done in polynomial-time due to \cite{MVC}. Therefore, by applying the method in Section \ref{Section:Approximation}, we can find a $\sfrac{3}{4}$-approximate CS of a given HVCG instance in polynomial-time. Note that this approximation ratio is better than what we can achieve for HEGs.

\bibliographystyle{abbrv}
\bibliography{references}

\newpage

\section*{Appendix}

\subsection*{Missing Proof of Observation \ref{Observation:Submodular}}

\begin{proof}
Let $X, Y \subseteq N$ be such that $X \subseteq Y$ and let $x \in N \setminus Y$. We first show that $U(X) \leq U(Y)$, which means that $U$ is monotone. We then show that $U(X \cup \{x\}) - U(X) \geq U(Y \cup \{x\}) - U(Y)$, which means that $U$ is submodular.

\paragraph{Monotonicity:} For a given skill $s \in S$, let $x^*$ and $y^*$ be the agents with maximum expertise on skill $s$, respectively in $X$ and $Y$. Since $X \subseteq Y$, $e_{x^*}(s) \leq e_{y^*}(s)$.   Then, $E_X(s) \leq E_Y(s)$ for all skills $s \in S$. Therefore, $\sum_{s \in S} E_X(s) \leq \sum_{s \in S} E_Y(s)$, which means that $U(X) \leq U(Y)$ by definition.

\paragraph{Submodularity:} Notice that $E_{X \cup \{x\}}(s) = \max \{e_x(s), E_X(s)\}$ for all skills $s \in S$ by definition. Therefore,
$U(X \cup \{x\}) - U(X) = \sum_{s \in S} \max \{e_x(s) - E_X(s), 0\}$ as we show below.
\begin{align*}
U(X \cup \{x\}) - U(X) &= \sum_{s \in S} E_{X \cup \{x\}}(s) - \sum_{s \in S} E_X(s) \\
&= \sum_{s \in S} \max \{e_x(s), E_X(s)\} - \sum_{s \in S} E_X(s) \\
&= \sum_{s \in S} \max \{e_x(s), E_X(s)\} - E_X(s) \\
&= \sum_{s \in S} \max \{e_x(s) - E_X(s), 0\}
\end{align*}

Notice that similarly $U(Y \cup \{x\}) - U(Y) = \sum_{s \in S} \max \{e_x(s) - E_Y(s), 0\}$. Recall that $E_X(s) \leq E_Y(s)$ for all skills $s \in S$. Then, $\max \{e_x(s) - E_X(s), 0\} \geq  \max \{e_x(s) - E_Y(s), 0\}$ for all skill $s \in S$. Therefore, $U(X \cup \{x\}) - U(X) \geq U(Y \cup \{x\}) - U(Y)$. \qed
\end{proof}

\end{document}